\documentclass[11pt]{article}

\usepackage[letterpaper,margin=1in]{geometry}
\usepackage{xcolor,graphicx,amssymb,amsthm,amsmath,hyperref,url,cite,mathrsfs,stmaryrd}
\counterwithin{figure}{section}
\usepackage{thm-restate}
\usepackage[utf8]{inputenc}
\usepackage[T1]{fontenc}

\graphicspath{{./figures/}}

\newtheorem{theorem}{Theorem}[section]
\newtheorem{lemma}[theorem]{Lemma}
\newtheorem{corollary}[theorem]{Corollary}
\newtheorem{proposition}[theorem]{Proposition}

\newcommand{\emphdef}[1]{\textbf{\textit{#1}}}

\newcommand{\cZ}{\mathbb{Z}}
\newcommand{\cR}{\mathbb{R}}
\newcommand{\torus}{\mathbb{T}}
\title{Making Multicurves Cross Minimally on Surfaces}

\author{Loïc Dubois\thanks{LIGM, Université Gustave Eiffel, F-77454 Marne-la-Vallée, France. LORIA, Inria, Université de Lorraine, F-54000 Nancy, France. The author was supported by the grant ANR-19-CE40-0014 (MIN-MAX), of the French National Research Agency ANR.}}

\date{2024}

\begin{document}

\maketitle

\begin{abstract}
On an orientable surface $S$, consider a collection $\Gamma$ of closed curves. The (geometric) intersection number $i_S(\Gamma)$ is the minimum number of self-intersections that a collection $\Gamma'$ can have, where $\Gamma'$ results from a continuous deformation (homotopy) of $\Gamma$. We provide algorithms that compute $i_S(\Gamma)$ and such a $\Gamma'$, assuming that $\Gamma$ is given by a collection of closed walks of length $n$ in a graph $M$ cellularly embedded on $S$, in $O(n \log n)$ time when $M$ and $S$ are fixed.

The state of the art is a paper of Despré and Lazarus [SoCG 2017, J. ACM 2019], who compute $i_S(\Gamma)$ in $O(n^2)$ time, and $\Gamma'$ in $O(n^4)$ time if $\Gamma$ is a single closed curve. Our result is more general since we can put an arbitrary number of closed curves in minimal position. Also, our algorithms are quasi-linear in $n$ instead of quadratic and quartic, and our proofs are simpler and shorter.

We use techniques from two-dimensional topology and from the theory of hyperbolic surfaces. Most notably, we prove a new property of the reducing triangulations introduced by Colin de Verdière, Despré, and Dubois [SODA 2024], reducing our problem to the case of surfaces with boundary. As a key subroutine, we rely on an algorithm of Fulek and Tóth [JCO 2020].
\end{abstract}

\section{Introduction}

\paragraph*{The problem}On a surface $S$, a collection $\Gamma$ of closed curves is in \emphdef{general position} if no point of $S$ is the image of more than two points of $\Gamma$, and if every self-intersection of $\Gamma$ is a crossing. In addition, $\Gamma$ is in \emphdef{minimal position} if no continuous deformation, in other words no homotopy, can decrease its number of crossings. The (geometric) \emphdef{intersection number} $i_S(\Gamma)$ is the number of crossings of a curve in minimal position homotopic to $\Gamma$. In this paper, we provide algorithms for the following problems: given a collection $\Gamma$ of closed curves on a surface $S$, compute $i_S(\Gamma)$, and construct a collection $\Gamma'$ in minimal position, homotopic to $\Gamma$. First, we review previous works.

\paragraph*{Previous works}The intersection number of closed curves on surfaces was first studied by the mathematical community. The problem of determining if a given closed curve is homotopic to a simple (injective) one was studied by Poincaré~\cite{p-ccal-04}, more recently by Reinhart~\cite{r-ajccs-62}, Chillingworth~\cite{c-sccs-69,c-wns2-72}, and Birman and Series~\cite{bs-ascs-84}. Later, Cohen and Lustig~\cite{cl-pggin-87} and Lustig~\cite{l-pggin2-87} (see also Hass and Scott~\cite{hs-ics-85}) determined the intersection number of one or two closed curves. De Graaf and Schrijver~\cite{DEGRAAF1997134}, and independently Hass and Scott~\cite{hass1994shortening}, found out that closed curves can be put in minimal position via homotopy moves that never increase the number of crossings. 

It emerges from these studies that, among the orientable surfaces, those of genus $g \geq 2$ without boundary constitute the hardest cases. Also, if a collection $\Gamma'$ of closed curves is in minimal position, then every closed curve in $\Gamma'$, and every pair of closed curves in $\Gamma'$ is also in minimal position~\cite[Section~1.2.4]{fm-pmcg-12}. Thus, computing $i_S(\Gamma)$ for an arbitrary $\Gamma$ boils down to computing $i_S(\Gamma)$ when $\Gamma$ contains at most two closed curves. However, an algorithm that can put one or two closed curves in minimal position would not necessarily extend to an arbitrary number of closed curves.

Recently, these problems have been revisited by computational topologists. The current state of the art is an algorithm provided by Despré and Lazarus~\cite{dl-cginc-19} in a technical paper of 49 pages. In the (popular) model they use, the input curves in $\Gamma$ are specified up to homotopy by a collection of closed walks $C$ in a graph $M$ \emphdef{cellularly embedded} on $S$, in that the faces of $M$ all have genus zero and at most one boundary component~\footnote{There exist standard data structures to represent combinatorial maps of cellular graph embeddings on orientable surfaces~\cite{e-dgteg-03,k-ugpdd-99}.}. They proved:

\begin{theorem}[Despré, Lazarus, 2019]\label{thm:despre-lazarus}
Let $M$ be a graph of size $m$ cellularly embedded on an orientable surface $S$. Let $C$ be a collection of either one or two closed walks of total length $n$ in $M$. One may compute $i_S(C)$ in $O(m + n^2)$ time. 
\end{theorem}

When $S$ has negative Euler characteristic, Despré and Lazarus also provide an algorithm to compute a closed curve $\gamma'$ in minimal position, homotopic to a given closed walk $C$. They derive from $M$ a quadrangulation $Q$, and return $\gamma'$ as an ``infinitesimal perturbation'' of a closed walk $C'$ in $Q$: in this paper, we say that $\gamma'$ is a \emphdef{perturbation} of $C'$, and we leave this notion informal for now. They proved:

\begin{theorem}[Despré, Lazarus, 2019]\label{thm:despre-lazarus-2}
Let $M$ be a graph of size $m$ cellularly embedded on an orientable surface $S$ of negative Euler characteristic~\footnote{Despré and Lazarus~\cite[Theorem~2]{dl-cginc-19} implicitly assume that the Euler characteristic of the surface is negative. Indeed, their algorithm computes a ``system of quads''~\cite[Section~8]{dl-cginc-19}, a quadrangulation defined by them only when the Euler characteristic is negative~\cite[Section~4.2]{dl-cginc-19}.}. Let $C$ be a closed walk of length $n$ in $M$. One may construct in $O(m + n^4)$ time a quadrangulation $Q$, a closed walk $C'$ of length $O(n)$ in $Q$, homotopic to $C$, and a perturbation of $C'$ with $i_S(C)$ self-crossings.
\end{theorem}

We insist that Theorem~\ref{thm:despre-lazarus-2} does not cope with more than one closed walk. Although Despré and Lazarus do not mention it, their output can easily be turned by isotopy into a perturbed closed walk of length $O(mn)$ in $M$ (instead of $Q$), at an additional cost of $O(mn)$ time.

A simpler problem is to consider only the perturbations of the input collection $C$ of closed walks, instead of searching over all the collections $C'$ homotopic to $C$, and to look for one with minimum self-crossing. This problem was studied by Fulek and Tóth~\cite{fulek2020crossing} when the closed walks in $C$ have no \emph{spur}, that is, when they never take an edge of $M$ and its reversal consecutively. They proved:

\begin{theorem}[Fulek, Tóth, 2020]\label{fulek}
Let $M$ be an embedded graph of size $m$. Let $C$ be a collection of closed walks of length $n$ in $M$, without spur. One may construct a minimal perturbation of $C$ in $O(m + n \log n)$ time.
\end{theorem}

The result of Fulek and Tóth has a different statement in their paper~\cite[Theorem~1]{fulek2020crossing}. In particular, the embedded graph lies in the plane, and only one closed walk is given as input. Nevertheless, Theorem~\ref{fulek} follows from their work (see Appendix~\ref{app:fulek}).

The problem of deciding whether a drawing of a graph on an orientable surface can be untangled, in other words, whether it is homotopic to an injective map, was studied by Colin de Verdière, Despré, and Dubois~\cite{untangling-graphs}. They provided a polynomial time algorithm. Along the way, they introduced \emphdef{reducing triangulations}. Every orientable surface of genus $g \geq 2$ without boundary admits a reducing triangulation. A triangulation is reducing if its vertices all have degree greater than or equal to eight, and if its dual graph is bipartite.  Reducing triangulations support \emphdef{reduced walks} and \emphdef{reduced closed walks}. Those are unique among their homotopy class, are stable upon subwalk and reversal, and can be computed in linear time from any given (closed) walk. Informally, they are convenient discrete analogs of geodesics on hyperbolic surfaces.

\paragraph*{Our results} We simplify, improve, and generalize the results of Despré and Lazarus. Naturally, we use the same model for the input: a collection $C$ of closed walks in a graph $M$ cellularly embedded on an orientable surface $S$. Our presentation focuses on the surfaces of genus $g \geq 2$ without boundary, as they constitute the hardest cases, but we prove similar results (with improved complexities and simpler proofs) for the surfaces with boundary and the torus. The case of the sphere is trivial. Here is the main result of the paper:

\begin{theorem}\label{thm:main-without-boundary}
Let $M$ be a graph of size $m$ cellularly embedded on an orientable surface $S$ of genus $g \geq 2$ without boundary. Let $C$ be a collection of closed walks of length $n$ in $M$. One may compute $i_S(C)$ in $O(m + g^2 + gn\log(gn))$ time. One may construct in $O(g^2 mn + gn \log gn)$ time a collection $C'$ of closed walks of length $O(g^2mn)$ in $M$, homotopic to $C$, and a perturbation of $C'$ with $i_S(C)$ self-crossings .
\end{theorem}

We even deduce from Theorem~\ref{thm:main-without-boundary} an algorithm of improved time complexity by allowing for a more compact representation of the output curve (see Corollary~\ref{cor:main-corollary}), such as when Despré and Lazarus return in a quadrangulation instead of $M$. Compared to Despré and Lazarus, our result is more general since we can put an arbitrary number of closed curves in minimal position. Also, our algorithms are quasi-linear in $n$ instead of quadratic and quartic, and our proofs are simpler and shorter as we benefit from two recent tools that were not available to them: the result of Fulek and Tóth, and the reducing triangulations of Colin de Verdière, Despré, and Dubois. We highlight the relationship between our problem and those tools. 

\paragraph*{Overview of the paper}After some preliminaries in Section~\ref{sec:preliminaries}, our starting point is the result of Fulek and Tóth, Theorem~\ref{fulek}. We observe (see Lemma~\ref{lem:perturbations}) that in the setting of Theorem~\ref{fulek}, if every face of $M$ contains a boundary component of $S$, then the returned perturbation of $C$ has exactly $i_S(C)$ self-crossings (not more). Motivated by this observation, we try to reduce our problem to an application of Theorem~\ref{fulek}. The difficulty is that, while spurs can trivially be eliminated from the input closed walks, the faces of $M$ do not generally contain a boundary component of $S$. We overcome this difficulty by proving in Section~\ref{sec:redux} the following new property of reducing triangulations:

\begin{restatable}{proposition}{keyprop}\label{prop:redux}
Let $T$ be a reducing triangulation of an orientable surface $S$ of genus $g \geq 2$ without boundary. Let $\Sigma$ be the surface obtained from $S$ by removing an open disk from each face of $T$. Let $C$ be a collection of reduced closed walks in the 1-skeleton of $T$. Then $i_S(C) = i_{\Sigma}(C)$.
\end{restatable}

Proposition~\ref{prop:redux} is largely inspired by~\cite[Proposition~4.3]{untangling-graphs}. In Section~\ref{sec:theorems}, we combine Proposition~\ref{prop:redux} with the result of Fulek and Tóth (Theorem~\ref{fulek}) to prove our main theorem, Theorem~\ref{thm:main-without-boundary}. Essentially, we observe that Theorem~\ref{thm:main-without-boundary} is straightforward (with improved complexities even) when $M$ is a reducing triangulation, and we handle the other cases with conversions between models adapted from~\cite[Section~7]{untangling-graphs} (informally, we transform $M$ into a reducing triangulation $T$, push $C$ by homotopy in $T$, compute in $T$, and pull back the result in $M$). In the same section, we handle the surfaces with boundary and the torus. Again, we reduce to an application of Theorem~\ref{fulek}, but reducing triangulations are not needed anymore. For the surfaces with boundary, we only perform the conversions between models adapted from~\cite[Section~8]{untangling-graphs}. For the torus, we define a particular kind of closed walks on a one-vertex embedded graph, for which we can prove a property similar to Proposition~\ref{prop:redux} (see Proposition~\ref{prop:key-torus}).

\section{Preliminaries}\label{sec:preliminaries}

We assume some familiarity with basic notions of graph theory, and of the topology of surfaces. In this paper, every surface has finite type and is orientable, without further mention.

\subsection{Lifting closed curves in the universal cover of a surface}

We will use standard notions of the theory of covering spaces, applied to surfaces. For more details, see, e.g.,~\cite[Section~10.4]{basic-topology}. 

Let $S$ be a surface of genus $g \geq 1$ without boundary. The \emphdef{universal cover} of $S$ is a surface $\widetilde S$ homeomorphic to an open disk, equipped with a local homeomorphism $\pi : \widetilde S \to S$. A \emphdef{lift} of a point $x \in S$ is a point $\widetilde x \in \widetilde S$ such that $\pi(\widetilde x) = x$. Given a closed curve $c : \cR/\cZ\to S$, let $c':\cR\to S$ be such that $c'(t)=c(t \text{ mod } 1)$.  A \emphdef{lift} of $c$ is a map $\widetilde c : \cR \to \widetilde S$ such that $\pi \circ \widetilde c = c'$. Given any $t \in \cR$ and any lift $\widetilde x$ of $c'(t)$, there is a lift $\widetilde c$ of $c$ that satisfies $\widetilde c(t) = \widetilde x$. We usually identify two lifts $\widetilde c, \widetilde d : \cR \to \widetilde S$ of the same closed curve if they differ by an integer translation, that is if there is $k \in \cZ$ such that $\widetilde c(t) = \widetilde d(t +k)$ for every $t \in \cR$. 

\subsection{Hyperbolic surfaces}\label{sec:background-hyperbolic}

We will use standard notions of two-dimensional hyperbolic geometry. See, e.g., Farb and Margalit~\cite[Chapter~1]{fm-pmcg-12}, also Cannon, Floyd, Kenyon, and Parry~\cite{cfkp-hg-97}. 

A \emphdef{hyperbolic surface} is a metric surface (here understood as a real, smooth two-dimensional manifold equipped with a Riemannian metric) locally isometric to the hyperbolic plane. Some hyperbolic surfaces can be constructed by considering \emphdef{ideal} hyperbolic polygons, the sides of which are complete (thus, of infinite length) geodesics, pairing the sides of the polygons, and identifying every two paired sides in a way that respects the orientations of the polygons. The topology of the resulting surface is that obtained by \emphdef{puncturing} (removing points from) a compact surface without boundary.

These surfaces enjoy the following specific properties: (1) There is a unique geodesic path homotopic to any given path; (2) there is a unique geodesic closed curve freely homotopic to any given closed curve, provided that curve is non-contractible and not homotopic to a neighborhood of a puncture.

\subsection{Reducing triangulations}

This section closely follows the presentation of~\cite[Section~3]{untangling-graphs}. A triangulation $T$ is reducing if its vertices have degree greater than or equal to eight, and if its dual graph is bipartite. Every surface of genus $g \geq 2$ without boundary admits a reducing triangulation. Reducing triangulations support \emphdef{reduced walks} and \emphdef{reduced closed walks}. We shall not need their definitions, but we give an informal account for completeness. Color the triangles of $T$ red and blue so that adjacent triangles receive distinct colors. In the 1-skeleton of $T$, a walk makes a bad turn if it uses an edge and its reversal consecutively. The walk also makes a bad turn if when going over an internal vertex it leaves only one triangle on its left, or if it leaves two triangles on its left but the first one encountered is red. Finally, a walk makes a bad turn if its reversal makes a bad turn. A walk is reduced if it makes no bad turn. A closed walk $C$ is reduced if $C$ makes no bad turn, with the following exception. If, up to reversing it, $C$ leaves three triangles on its left at every vertex, the middle one being blue, then $C$ is not reduced as a closed walk. Clearly, the reversal, and any subwalk of a reduced (closed) walk is reduced. Importantly:

\begin{lemma}[\cite{untangling-graphs}, Proposition~3.1 and Proposition~3.4]\label{lem:unicity}
In a reducing triangulation $T$, any two homotopic reduced walks are equal. Also, any two freely homotopic, non-contractible, reduced closed walks are equal. 
\end{lemma}

\begin{lemma}[\cite{untangling-graphs}, Proposition~3.7]\label{lem:reduce}
Let $C$ be a closed walk of length $n$ in a reducing triangulation $T$. One may compute a reduced closed walk freely homotopic to $C$ in $O(n)$ time.
\end{lemma}

Lemma~\ref{lem:unicity} will be used in the proof of Proposition~\ref{prop:redux}. Lemma~\ref{lem:reduce} will be used in the proof of Theorem~\ref{thm:main-without-boundary}. A statement similar to Lemma~\ref{lem:reduce} holds for reduced walks (the ones not closed), but we will not need it.

\subsection{Patch systems and perturbations}\label{sec:perturbations}

Let $M$ be a graph cellularly embedded on a surface $S$. We present an adaptation of the \emphdef{strip systems} used by Akitaya, Fulek, and Tóth~\cite{aft-rweg-19} (and others, see .e.g~\cite{CORTESE20091856}). Our presentation follows the one of Colin de Verdière, Despré, and Dubois~\cite[Section~2.2]{untangling-graphs}. See Figure~\ref{fig:patch-system}. The \emphdef{patch system} $\Sigma$ of $M$ is a surface that can be obtained from $S$ by first filling up any boundary component of $S$ (attaching a closed disk to it), and then by removing an open disk from the interior of every face of $M$. The patch system $\Sigma$ is a surface with boundary homotopically equivalent to $M$, usually thought of as a ``closed neighborhood'' of $M$ in $S$. By construction, the boundary components of $\Sigma$ correspond to the faces of $M$. We equip $\Sigma$ with a collection of pairwise-disjoint simple arcs between its boundary components: one arc $a$ for each edge $e$ of $M$, where $a$ crosses $e$ once and does not intersect $M$ anywhere else. Those arcs divide the interior of $\Sigma$ into open disks, one for each vertex of $M$.

\begin{figure}[ht]
    \centering
    \includegraphics[scale=0.25]{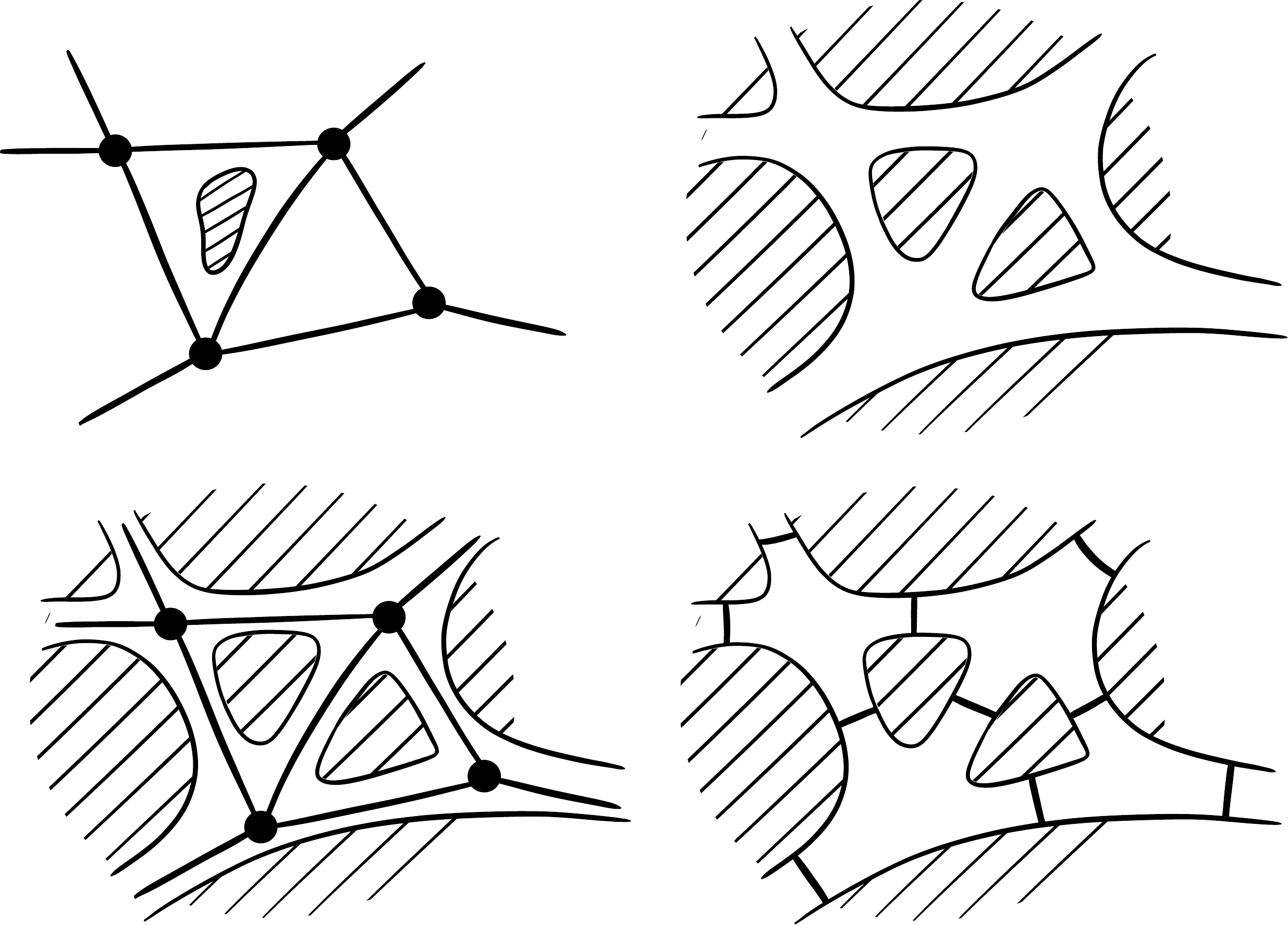}
    \caption{(Top Left) An embedded graph $M$. (Top Right) The patch system $\Sigma$ of $M$. (Bottom Left) We usually think of $\Sigma$ as a ``closed neighborhood'' of $M$. (Bottom Right) The arcs of $\Sigma$.}
    \label{fig:patch-system}
\end{figure}

Let $\gamma$ be a closed curve in general position in the patch system $\Sigma$ of $M$, where $\gamma$ intersects every disk of $\Sigma$ as a collection of simple paths, and where every two such paths intersect at most once in the disk. We retain from $\gamma$ only two things: the sequence of arcs of $\Sigma$ crossed by $\gamma$ and, for every arc $a$ of $\Sigma$, the order in the crossings between $\gamma$ and $a$ occur along $a$. Dually, we retain a closed walk $C$ in $M$ and, for every edge $e$ of $M$, a linear order $\prec_e$ on the occcurences of $e$ in $C$. In this paper, we say that $(\prec_e)_e$ is a \emphdef{perturbation} of $C$, and that $C$ and $(\prec_e)_e$ constitute a \emphdef{perturbed closed walk}. Several non-isotopic closed curves $\gamma$ may be represented by the same perturbed closed walk but they all have the same number of self-crossings, making the distinction between them irrelevant to us. This paragraph trivially extends to collections of closed curves. 

\begin{lemma}\label{lem:perturbations}
Let $\Sigma$ be the patch system of a graph $M$ cellularly embedded on a surface $S$. Let $C$ be a collection of closed walks in $M$. Let $\Gamma$ be a perturbation of $C$ with a minimum number of self-crossings. If $C$ has no spur, then $\Gamma$ has $i_{\Sigma}(C)$ self-crossings.
\end{lemma}

\begin{figure}[ht]
    \centering
    \includegraphics[scale=0.2]{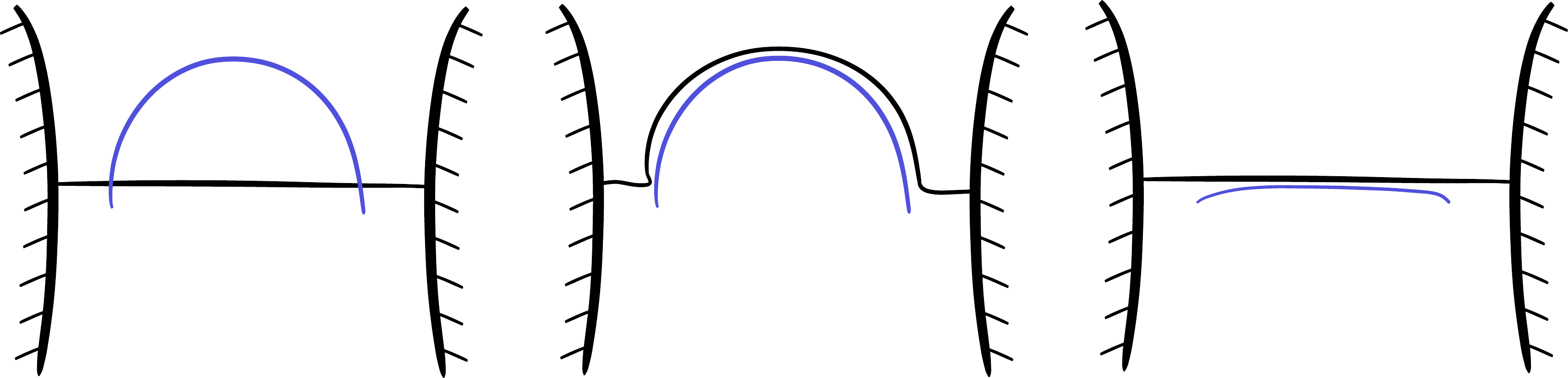}
    \caption{Removing a bigon in the proof of Lemma~\ref{lem:perturbations}.}
    \label{fig:remove-bigons}
\end{figure}

\begin{proof}
We assume for clarity that $C$ is a single closed walk, but the proof trivially extends to any collection of closed walks. Let $\gamma$ be a closed curve in general position in $\Sigma$, homotopic to $C$, with $i_S(C)$ self-crossings. Without loss of generality, assume that $\gamma$ intersects the arcs of $\Sigma$ a minimum number of times (among the closed curves that match its definition). Record the arcs crossed by $\gamma$ by the corresponding closed walk $C'$ in $M$, so that $\gamma$ is a perturbation of $C'$. We shall prove $C = C'$ (up to cyclic permutation). To do so, it is enough to prove that $C'$ has no spur. Indeed, $C$ and $C'$ are homotopic in $\Sigma$, and $\Sigma$ is homotopy equivalent to $M$, so $C$ and $C'$ are homotopic in $M$. Moreover, in a graph, every two homotopic closed walks without spur are equal, see~\cite[Chapter~2]{s-ctcgt-93}. 

We prove that $C'$ has no spur by contradiction, so assume that $C'$ has a spur. There is a portion of $\gamma$ that crosses an arc $a$ of $\Sigma$ and then crosses $a$ again consecutively, in the opposite direction. See Figure~\ref{fig:remove-bigons}. Let $\gamma_0$ and $a_0$ be the respective portions of $\gamma$ and $a$ between the two crossing points. Let $n$ and $m$ be the number of crossings of $\gamma$ with the interiors of $\gamma_0$ and $a_0$, respectively. Redrawing $\gamma_0$ in a neighborhood of $a_0$ would not decrease the number of self-crossings of $\gamma$, so $n \leq m$. Redrawing $a_0$ in a neighborhood of $\gamma_0$ would not decrease the number of crossings between $a_0$ and $\gamma$, so $m+2 \leq n$. This is a contradiction. 
\end{proof}

\begin{lemma}\label{lem:compute-nb-crossings}
Let $M$ be a graph of size $m$ cellularly embedded on a surface $S$. Let $\Sigma$ be the patch system of $M$. Let $C$ be a collection of closed walks of length $n$ in $M$. Let $\Gamma$ be a perturbation of $C$. One may compute the number of self-crossings of $\Gamma$ in $O(m + n \log n)$ time.
\end{lemma}

\begin{proof}
One must count the number of self-crossings of $\Gamma$ within each disk of the patch system $\Sigma$. There is an immediate reduction to the following problem. Fix $n \geq 2$. Consider a matching of the set $[n]$, represented by its set $E$ of edges. Write every edge $e \in E$ as $e = ij$, where $i$ and $j$ are the vertices of $e$, and $i < j$. Say that an edge $ij \in E$ crosses an other edge $uv \in E$ if $i<u<j<v$, up to exchanging $ij$ with $uv$. Let $\bot E$ be the number of unordered pairs of edges of $E$ that cross. Our problem is to compute $\bot E$. Assuming that $E$ is a perfect matching, we claim that we can compute $\bot E$ in $O(n \log n)$ time.

To prove this claim, we apply a divide-and-conquer strategy. The base cases (small values of $n$) are trivial. In general, we consider the unique $k \in [n]$ for which exactly $\left\lceil n/2 \right\rceil$ edges $uv \in E$ satisfy $v \leq k$. We partition $E$ into three sets $E_0,E_1,E_2$ as follows. An edge $uv \in E$ belongs to $E_0$ if it satisfies $v \leq k$. It belongs to $E_1$ if it satisfies $k < u$. And it belongs to $E_2$ otherwise, that is if $u \leq k < v$. By definition of $k$, each one of the sets $E_0,E_1,E_2$ contains at most $\left\lceil n/2 \right\rceil$ edges. For every $l \leq k$, we let $\omega(l)$ be the number of edges $uv \in E_0$ that satisfy $u < l < v$. For every $l > k$, we let $\omega(l)$ be the number of edges $uv \in E_1$ that satisfy $u < l < v$. We have the recursion formula $\bot E = \bot E_0 + \bot E_1 + \bot E_2 + \sum_{uv \in E_2} (\omega(u) + \omega(v))$. We use this formula in our recursion step, as follows. First, we compute $k$ in $O(n)$ time. Then, we compute $E_0, E_1,$ and $E_2$ in $O(n)$ time. Also, we compute $\omega$ on $[n]$ in $O(n)$ time by dynamic programming.  Finally, we recurse to compute $\bot E_0 , \bot E_1$, and $\bot E_2$. And we deduce $\bot E$ in $O(n)$ time from the recursion formula.     
\end{proof}

\section{A property of reducing triangulations}\label{sec:redux}

In this section we prove Proposition~\ref{prop:redux}, which we recall:

\keyprop*

The proof uses classical arguments, and adapts the proof of~\cite[Proposition~4.3]{untangling-graphs}. First, we review the necessary background.

\subsection{Background}

In this section, given two closed curves $c$ and $d$ on a surface $S$, we let $i^*_S(c,d)$ be the minimum number of crossings between any two closed curves in general position, homotopic to $c$ and $d$. We apply the strategy of Despré and Lazarus and focus on \emphdef{primitive} closed curves. A closed curve $c$ on $S$ is primitive if there is no $k \geq 2$ for which $c$ would be homotopic to the $k^{th}$ power of another closed curve in $S$ (where the $k$-th power of closed curve $d : \cR / \cZ \to S$ is the closed curve $t \to d(kt)$). The following relates the intersection number of arbitrary closed curves to the intersection number of the primitive ones (see also~\cite[Proposition~26]{dl-cginc-19}):

\begin{lemma}[Theorems~1-6-7 in \cite{DEGRAAF1997134}]\label{lem:degraaf}
Let $S$ be a surface of genus $g \geq 1$ without boundary. Let $c$ and $d$ be two closed curves on $S$ homotopic to respectively the $n^{th}$ power of a primitive closed curve $\hat c$, and the $m^{th}$ power of a primitive closed curve $\hat d$, for some $n,m \geq 1$. Then $i_S(c) = n^2 i_S(\hat c) + n-1$. If $\hat c$ (or the reversal of $\hat c$) is homotopic to $\hat d$, then $i^*_S(c,d) = 2nm \times i_S(\hat c)$. Otherwise, $i^*_S(c,d) = nm \times i^*_S(\hat c,\hat d)$.
\end{lemma}

In the setting of Lemma~\ref{lem:degraaf}, draw $c$ and $d$ in the neighborhoods of $\hat c$ and $\hat d$ in the way described by Figure~\ref{fig:non-primitive-curves}. Lemma~\ref{lem:degraaf} essentially says that if $\hat c$ and $\hat d$ cross themselves and each other minimally, then so do $c$ and $d$.

\begin{figure}[ht]
    \centering
    \includegraphics[scale=0.28]{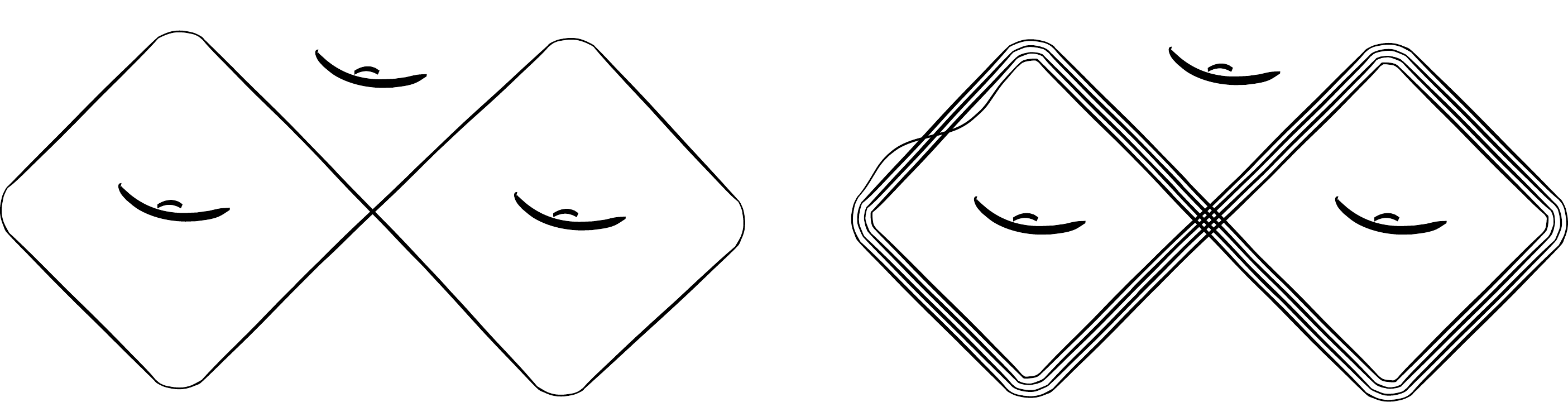}
    \caption{(Left) A primitive closed curve $\hat c$, in minimal position. (Right) The fourth power $c$ of $\hat c$, in minimal position in a neighborhood of $\hat c$.}
    \label{fig:non-primitive-curves}
\end{figure}

The following gives a sufficient condition for two primitive closed curves to cross each other minimally. It is a classical fact, very similar to \cite[Lemma~4]{dl-cginc-19}. We could not find a proof of this exact statement in the literature, so we provide one in Appendix~\ref{app:minimal-position-2} for completeness. The most similar result we could find in the litterature is by Hass and Scott~\cite[Theorem~3.5]{hs-ics-85}, but they stated it only for a single closed curve, and we need to handle two closed curves.

\begin{lemma}\label{lem:minimal-position-2}
On a surface $S$ of genus $g \geq 2$ without boundary, consider two primitive closed curves $c$ and $d$ (possibly homotopic) in general position. If in the universal cover of $S$, no lift of $c$ crosses a lift of $d$ more than once, then $c$ and $d$ cross each other minimally.
\end{lemma}

The following gives a sufficient condition for a single primitive closed curve to cross itself minimally. It is easily derived from Lemma~\ref{lem:degraaf} and Lemma~\ref{lem:minimal-position-2}. It is also an immediate consequence of the result of Hass and Scott~\cite[Theorem~3.5]{hs-ics-85}.

\begin{lemma}\label{lem:minimal-position-1}
On a surface $S$ of genus $g \geq 2$ without boundary, consider a primitive closed curve $c$ in general position. If in the universal cover of $S$, the lifts of $c$ are injective and do not cross more than once, then $c$ crosses itself minimally.
\end{lemma}

\begin{proof}[Proof of Lemma~\ref{lem:minimal-position-1}]
Let $p$ be the number of self-crossings of $c$. Consider two parallel copies $c_1$ and $c_2$ of $c$ in general position in a neighborhood of $c$. By construction, there are $2p$ crossings between $c_1$ and $c_2$ (there would be $4p$ crossings if we also counted the self-crossings of $c_1$ and those of $c_2$). A lift of $c_1$ cannot cross a lift of $c_2$ more than once, so $2p = i^*_S(c_1, c_2)$ by Lemma~\ref{lem:minimal-position-2}. We have $i^*_S(c_1, c_2) = 2 i_S(c)$ by Lemma~\ref{lem:degraaf}, since $c$ is primitive.
\end{proof}

\subsection{Proof of Proposition~\ref{prop:redux}}

\begin{proof}[Proof of Proposition~\ref{prop:redux}]
We have $i_S(C) \leq i_{\Sigma}(C)$ by the inclusion $\Sigma \subset S$. Let us prove the converse. Assume without loss of generality that no closed walk in $C$ is trivial (a closed walk is trivial if it consists of a single vertex). Then every $c \in C$ is non contractible on $S$ by Lemma~\ref{lem:unicity}, and is thus homotopic in $S$ to the $n(c)^{th}$ power power of a primitive closed curve, say $d$, for some $n(c) \geq 1$. Let $\hat{c}$ be the reduced closed walk homotopic to $d$, which exists by Lemma~\ref{lem:reduce}. Lemma~\ref{lem:unicity} implies that $c$ is actually equal to the $n(c)^{th}$ power of $\hat{c}$, since the two are freely homotopic non-contractible reduced closed walks. Let $\widehat{C}$ contain the resulting primitive reduced closed walks, counted without multiplicity. We shall construct, for every $\hat{c} \in \widehat{C}$, a closed curve $\widehat{\gamma}_{\hat{c}}$ homotopic to $\hat{c}$ in $\Sigma$, such that the collection $\{\widehat{\gamma}_{\hat{c}} \;|\; \hat{c} \in \widehat{C}\}$ is in general position. We claim the existence of such a collection for which $\widehat{\gamma}_{\hat{c}}$ crosses itself $i_S(\hat{c})$ times for every $\hat{c} \in \widehat{C}$, and for which $\widehat{\gamma}_{\hat{c}}$ and $\widehat{\gamma}_{\hat{d}}$ cross each other $i^*_S(\hat{c},\hat{d})$ times for every $\hat{c} \neq \hat{d} \in \widehat{C}$. Let us first explain why the claim infers the result. Realize every $c \in C$ by a closed curve $\gamma_c$ in the neighborhood of $\widehat{\gamma}_{\hat{c}}$ such in Figure~\ref{fig:non-primitive-curves}. Then, by Lemma~\ref{lem:degraaf}, the number of self-crossings of $\gamma_c$ is minimum among its homotopy class on $S$, and the number of crossings between $\gamma_c$ and $\gamma_d$ is also minimum.

To prove the claim, endow the interior of $\Sigma$ with a complete hyperbolic metric for which the arcs of $\Sigma$ are complete geodesics, as described in Section~\ref{sec:background-hyperbolic}. For every $\hat{c} \in \widehat{C}$, let $\widehat{\gamma}_{\hat{c}}$ be the geodesic closed curve in the homotopy class of $\hat{c}$ in $\Sigma$. We now see $\widehat{\gamma}_{\hat{c}}$ as a closed curve on $S$, by the inclusion $\Sigma \subset S$. The claim is immediate from Lemma~\ref{lem:minimal-position-2}, Lemma~\ref{lem:minimal-position-1}, and the following observation: in the universal cover of $S$, the lifts of the closed curves in $\{\widehat{\gamma}_{\hat{c}} \;|\; \hat{c} \in \widehat{C}\}$ are injective, and every two of them cannot cross more than once. We prove this observation by contradiction, so assume that one of those lifts is not injective. It contains a loop, which projects to a geodesic loop $\ell$ on $\Sigma$ that is contractible on $S$. The sequence of crossings of $\ell$ with the arcs of $\Sigma$ is that of a reduced walk. Thus, by Lemma~\ref{lem:unicity}, the loop $\ell$ does not cross any arc of $\Sigma$, which is impossible since $\ell$ is geodesic. Similarly, if two distinct lifts intersect twice without overlapping, then they form two paths with the same endpoints and otherwise disjoint, which project to geodesic paths $p$ and $q$ in $\Sigma$ that are homotopic in $S$. The sequence of crossings of $p$ and $q$ with the arcs of $\Sigma$ must be the same by Lemma~\ref{lem:unicity}, so $p$ and $q$ are homotopic in $\Sigma$, which is impossible since $p$ and $q$ are geodesics.
\end{proof}

\section{The algorithms}\label{sec:theorems}

\subsection{Surfaces without boundary}

In this section, we first prove Theorem~\ref{thm:main-without-boundary}. Then we derive in Corollary~\ref{cor:main-corollary} an algorithm of improved time complexity by allowing for a more compact representation of the output curve. 

Let us first prove Theorem~\ref{thm:main-without-boundary}. A conversion from graphs drawn on arbitrary cellular embeddings to graphs drawn on reducing triangulations was already described by Colin de Verdière, Despré and Dubois~\cite[Lemma~7.1]{untangling-graphs}. We reformulate their result in the context of closed walks (not general graph drawings):

\begin{lemma}[Particular case of Lemma~7.1 in~\cite{untangling-graphs}]\label{lem:conversion-redux}
Let $M$ be a graph of size $m$ cellularly embedded on a surface $S$ of genus $g \geq 2$ without boundary. Let $C$ be a collection of closed walks of length $n$ in $M$. One may construct in $O(m + g^2 + gn)$ time a reducing triangulation $T$ of $S$, and a collection $C'$ of closed walks of length $O(gn)$ in $T$, homotopic to $C$.
\end{lemma}

This lemma is all we need to compute intersection numbers. However, in order to put curves in minimal position, we must construct an overlay between the input cellular embedding and a reducing triangulation. Such overlay is not explicitly given by~\cite[Lemma~7.1]{untangling-graphs}, but it can be extracted from the proof with some additional work.

On a surface $S$ of genus $g \geq 2$ without boundary, a \emphdef{canonical system of loops} $L$ is a set of pairwise disjoint simple loops with a common basepoint, such that, when cutting $S$ along $L$, we obtain a canonical polygonal schema, a 4g-gon whose boundary reads $a_1 b_1 \overline{a_1} \overline {b_1} \dots a_g b_g \overline{a_g} \overline{b_g}$ in this order, where $a_1, \dots, a_g$ denote the loops in $L$, and bar denotes reversal.

\begin{figure}[ht]
    \centering
    \includegraphics[scale=0.65]{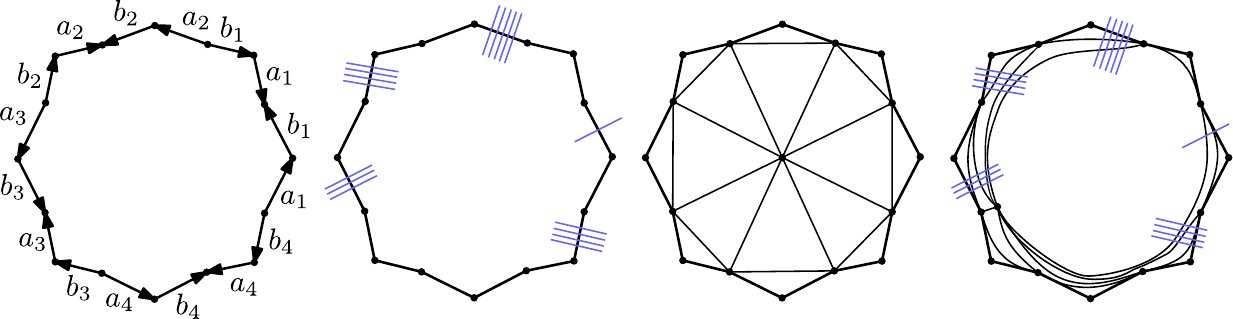}
    \caption{(Left) The canonical system of loops $L$ for the surface $S$ of genus $g=4$. (Middle Left) A blue graph $Q$ on $S$, in general position with respect to $L$, that intersects every edge of $L$ at most $m \geq 0$ times. (Middle Right) A reducing triangulation $T$ whose 1-skeleton contains $L$. (Right) An embedding of $T$ on $S$ such that each edge of $T$ crosses $Q$ at most $O(gm)$ times. This construction trivially generalizes to higher genus.}
    \label{fig:embed-redux}
\end{figure}

\begin{figure}[ht]
    \centering
    \includegraphics[scale=0.15]{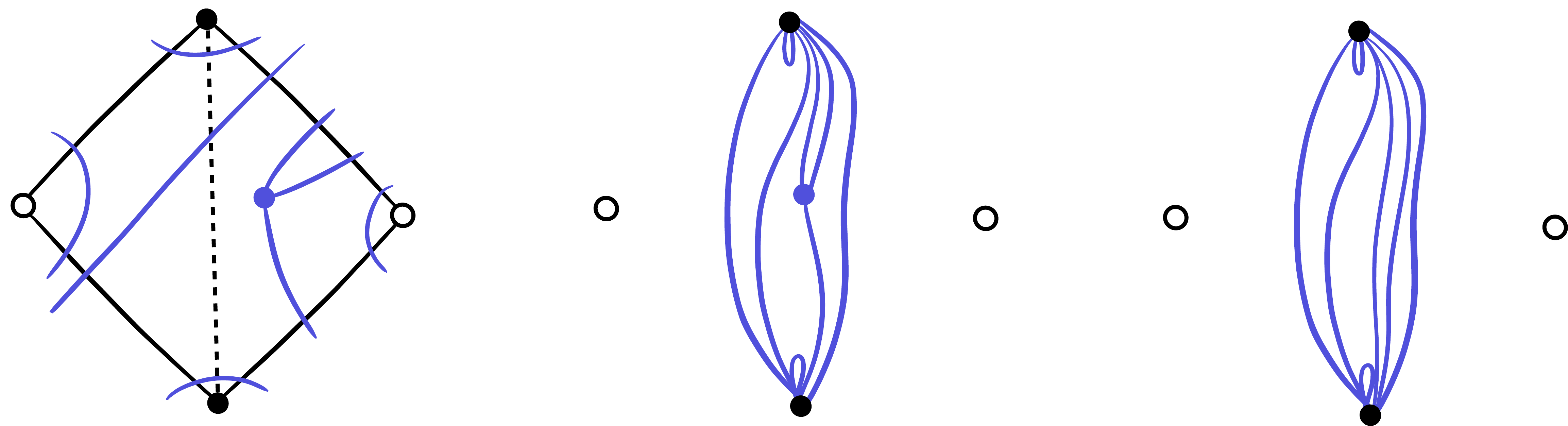}
    \caption{(Left) In the proof of Theorem~\ref{thm:main-without-boundary}, a portion of $T$ is here represented in blue in some face of $Q$. The two black disk vertices belong to $M$. The two circle vertices were inserted in faces of $M$ to build $Q$, they are dual vertices. The dashed edge was deleted from $M$ to build $Q$. The plain edges belong to $Q$. (From Left to Middle) In the overlay $Q'$ between $T$ and $Q$, every edge of $Q$ is detached from its incident dual vertex, then contracted. (From Middle to Right) Some edge incident to the blue disk vertex of $T$ is contracted.}
    \label{fig:contractions}
\end{figure}

\begin{proof}[Proof of Theorem~\ref{thm:main-without-boundary}]
First, we explain how to compute $i_S(C)$. Apply Lemma~\ref{lem:conversion-redux} to build in $O(m + g^2 + gn)$ time a reducing triangulation $T$ of $S$, and a collection $D$ of closed walks of length $O(gn)$ in $T$, homotopic to $C$. Apply Lemma~\ref{lem:reduce} to make the closed walks in $D$ reduced in $O(gn)$ time. Apply Theorem~\ref{fulek} to compute a minimal perturbation of $D$ in $O(gn \log (gn))$ time. This perturbation has $i_S(C)$ self-crossings by Proposition~\ref{prop:redux} and Lemma~\ref{lem:perturbations}. Apply Lemma~\ref{lem:compute-nb-crossings} to count them in $O(gn \log (gn))$ time.

Now, we explain how to construct a collection $C'$ of closed walks in $M$, homotopic to $C$, and a perturbation of $C'$ with $i_S(C)$ self-crossings. Let $Q$ be obtained from $M$ by inserting a new vertex in every face, by  adding an edge between this vertex and every corner of the face, and by removing from $M$ all its initial edges. Then $Q$ is a quadrangulation of size $O(m)$. Let $L$ be the canonical system of loops for the surface of genus $g$ without boundary. Let $T$ be a reducing triangulation whose 1-skeleton contains $L$ such as in~\cite[Figure~7.1]{untangling-graphs}. Apply a result of Pocchiola, Lazarus, Vegter, and Verroust~\cite{lpvv-ccpso-01} to compute in $O(gm)$ time an embedding of $L$ in general position with respect to $Q$, such that each loop of $L$ crosses each edge of $Q$ at most four times. Compute a collection $C'$ of reduced closed walks of length $O(gn)$ in $T$, homotopic to $C$. More precisely merge, for the needs of this paragraph only, the vertices of $Q$ with the vertex of $L$ so that every edge of $Q$ becomes a sequence of $O(g)$ arcs on the face of $L$ ($Q$ is not embedded anymore). Then push every arc into a path of length $O(1)$ in $T$. In the end, apply Lemma~\ref{lem:reduce} to make the closed walks in $C'$ reduced in $O(gn)$ time.

Apply Theorem~\ref{fulek} to compute in $O(gn \log(gn))$ time a minimal perturbation $\Gamma$ of $C'$. Then $\Gamma$ has $i_S(C)$ self-crossings by Proposition~\ref{prop:redux} and Lemma~\ref{lem:perturbations}.

Construct from the embedding of $L$ an embedding of $T$ in general position with respect to $Q$, such that each edge of $T$ crosses $Q$ at most $O(gm)$ times, see Figure~\ref{fig:embed-redux}. Build the overlay $Q'$ between $Q$ and the embedded $T$ in $O(g^2 m)$ time. Recall that, by construction, every vertex of $Q$ is either a vertex of $M$, or a dual vertex inserted in a face of $M$. Every edge $e$ of $Q$ is between a vertex $v$ of $M$ and a dual vertex $w$. In $Q'$, the edge $e$ is subdivided as a path $P$ between $v$ and $w$, defined by the intersection points between $e$ and $T$. Delete the edge of $P$ incident to $w$, and contract the other edges of $P$. Do so for every edge of $Q$. See Figure~\ref{fig:contractions}. The dual vertices are now isolated in $Q'$, delete them. Also, contract some of the edges incident to the vertices of $T$ so that every vertex of $Q'$ is now a vertex of $M$. At this point, $Q'$ is a sub-graph of $M$ except that some edges may now be multiple parallel edges, and some contractible loop-edges may have appeared. After all those edge contractions, $C'$ and $\Gamma$ constitute a perturbed collection of closed walks of length $O(g^2 mn)$ in $Q'$. Derive from them a perturbed collection of closed walks in $M$, not longer, and with no more self-crossing, by removing the loop-edges and merging the parallel edges.
\end{proof}

In the setting of Theorem~\ref{thm:main-without-boundary}, we get an algorithm of improved time complexity by first contracting a spanning tree in $M$ (modifying the input closed walks accordingly) to make $M$ have only one-vertex:

\begin{corollary}\label{cor:main-corollary}
Let $M$ be a one-vertex graph cellularly embedded on an orientable surface $S$ of genus $g \geq 2$ without boundary. Let $C$ be a collection of closed walks of length $n$ in $M$. One may construct a collection $C'$ of closed walks of length $O(g^3n)$ in $M$, homotopic to $C$, and a perturbation of $C'$ with $i_S(C)$ self-crossings, all in $O(g^3n + gn \log gn)$ time.
\end{corollary}

\begin{proof}
Construct a subgraph of $M$ by merging any two edges that bound a bigon, and by deleting any edge that bounds a monogon. This subgraph has size $O(g)$ by Euler formula. Apply Theorem~\ref{thm:main-without-boundary}.
\end{proof}

\subsection{Surfaces with boundary}

In this section, we prove results similar to Theorem~\ref{thm:main-without-boundary} and Corollary~\ref{cor:main-corollary} for the surfaces with boundary.

\begin{theorem}\label{thm:main-with-boundary}
Let $M$ be a graph of size $m$ cellularly embedded on a surface $S$ of genus $g \geq 0$ with $b \geq 1$ boundary components. Let $C$ be a collection of closed walks of length $n$ in $M$. One may construct a collection $C'$ of closed walks of length $O(mn)$ in $M$, homotopic to $C$, and a perturbation of $C'$ with $i_S(C)$ self-crossings, all in $O(mn + (g+b)n \log((g+b)n))$ time. 
\end{theorem}

Here, the conversion is given by~\cite[Lemma~8.4]{untangling-graphs}. We adapt the proof of this lemma.

\begin{proof}
Let $\widehat S$ be the surface without boundary obtained by filling up every boundary component of $S$ with a disk. In $\widehat S$, fix a spanning tree $Y$ of $M$, and a spanning tree $Y^*$ of the dual of $M$, so that $Y$ and $Y^*$ are disjoint. Let $L$ contain every edge of $M$ that does not belong to $Y$ and whose dual does not belong to $Y^*$. Let $M_0$ be the union of $Y$, of $L$, and of an additional set of edges whose duals belong to $Y^*$. Select those edges so that $M_0$ contains a boundary component of $S$ in every face, and is maximal under this requirement. Let $M_1$ result from $M_0$ by contracting $Y$. Then $M_1$ has size $O(g+b)$ by Euler formula. Push by homotopy in $S$ every edge of $C$ to a path in $M_0$. Then $C$ becomes a collection $C'$ of closed walks in $M_0$, without spur. Every edge of $C$ becomes in $C'$ a path that uses at most once every edge of $M_0$ (without loss of generality), so $C'$ has length $O(mn)$. Similarly, $C'$ corresponds to a collection $C_1$ of length $O((g+b)n)$ in $M_1$. Apply Theorem~\ref{fulek} to compute a minimal perturbation of $C_1$ in $O((g+b)n \log ((g+b)n)$ time. This perturbation has $i_S(C)$ self-crossings by Lemma~\ref{lem:perturbations}. Recover in $O(mn)$ time a perturbation of $C'$ without additional crossing.
\end{proof}

\begin{corollary}\label{cor:boundary}
Let $M$ be a one-vertex graph cellularly embedded on a surface $S$ of genus $g \geq 0$ with $b \geq 1$ boundary components. Let $C$ be a collection of closed walks of length $n$ in $M$. One may construct a collection $C'$ of closed walks of length $O((g+b)n)$ in $M$, homotopic to $C$, and a perturbation of $C'$ with $i_S(C)$ self-crossings, all in $O((g+b) n\log ((g+b)n))$ time. One may compute $i_S(C)$ at no extra cost.
\end{corollary}

\begin{proof}
Construct a subgraph of $M$ by merging any two edges that bound a bigon, and by deleting any edge that bounds a monogon. This subgraph has size $O(g+b)$ by Euler formula. Modify $C$ accordingly. Apply Theorem~\ref{thm:main-with-boundary} and Lemma~\ref{lem:compute-nb-crossings}.
\end{proof}

\subsection{Torus}

In this section, we prove results similar to Theorem~\ref{thm:main-without-boundary} and Corollary~\ref{cor:main-corollary} for the torus $\torus$. Given a collection $\Gamma$ of closed curves on $\torus$, there exist formulas for computing $i_{\torus}(\Gamma)$, see e.g.~\cite[Section~1.2.3]{fm-pmcg-12}. However, putting $\Gamma$ in minimal position in a purely discrete way requires additional work.

On a torus $\torus$, a \emphdef{canonical system of loops}~$L$ consists in two pairwise disjoint simple loops with common basepoint that cross at the basepoint. The dual embedded graph $L^*$ of $L$ is also a canonical system of loops on $\torus$. Endow $\torus$ with a flat metric for which the face of $L^*$ is isometric to the interior of a flat square (or parallelogram). Let $\torus^{\times}$ be obtained from $\torus$ by removing the vertex of $L^*$. We say that a closed walk $D$ in $L$ is \emphdef{quasi-geodesic} if $D$ has no spur, and if $D$ is homotopic to a (non-contractible) geodesic closed curve in $\torus^{\times}$. We insist that both the geodesic closed curve and the homotopy are in $\torus^{\times}$ (not in $\torus$).

\begin{lemma}\label{lem:quasi-geodesic}
On a torus $\torus$, let $L$ be a canonical system of loops. Let $C$ be a non-contractible closed walk of length $n$ in $L$. One may compute in $O(n)$ time a quasi-geodesic closed walk $D$ of length $O(n)$ in $L$, homotopic to $C$ in $\torus$.
\end{lemma}

\begin{proof}
Let $L^*$ be the canonical system of loops dual to $L$ on $\torus$. Let $\ell_1, \ell_2$ be the two loops of $L$, and let $\ell_1^*, \ell_2^*$ be their respective dual loops. Identify $\torus$ with the quotient $\cR^2 / \cZ^2$, such that $L^*$ lifts to the following grid: the vertex of $L^*$ lifts to $\cZ^2$, the loop $\ell_1^*$ lifts to the vertical segments between $(i,j)$ and $(i,j+1)$ for $i,j \in \cZ$, and the loop $\ell_2^*$ lifts to the horizontal segments between $(i,j)$ and $(i+1,j)$ for $i,j \in \cZ$. Orient $\ell_1$ and $\ell_2$ so that, in $\cR^2$, the lifts of $\ell_1$ cross the lifts of $\ell_1^*$ from left to right, and the lifts of $\ell_2$ cross the lifts of $\ell_2^*$ from bottom to top. Let $\torus^{\times}$ be the surface obtained from $\torus$ by removing the vertex of $L^*$. Let $\widetilde L^*$ and $\widetilde L$ be the lifts of $L^*$ and $L$.

First, we define the quasi-geodesic closed walk $D$, without actually computing it. For every $i \in \{1,2\}$, let $k_i$ record the number of times $C$ takes the loop $\ell_i$ in the positive direction, minus the number of times $C$ takes $\ell_i$ in the negative direction. Then $(k_1,k_2) \neq (0,0)$ since $C$ is non-contractible in $\torus$. Consider any point $p = (p_1,p_2) \in \cR^2$ for which the geodesic line $\widetilde \gamma$ containing $p$ and $p' := (p_1+k_1,p_2+k_2)$ does not intersect $\cZ^2$. Then $\widetilde \gamma$ projects to a geodesic closed curve $\gamma$ on $\torus^{\times}$, homotopic to $C$ in $\torus$. Record the sequence of crossings of $\widetilde \gamma$ with the edges of $\widetilde L^*$ between the points $p$ and $p'$. Dually, record a walk $\widetilde D$ in $\widetilde L$. Then $\widetilde D$ projects to a closed walk $D$ on $L$. Moreover, $D$ has no spur and is homotopic to $\gamma$ in $\torus^{\times}$. Thus, $D$ is a quasi-geodesic closed walk, homotopic to $C$ in $\torus$. Also, $D$ has length $O(k_1 + k_2) = O(n)$.

To compute $D$, one must choose a point $(p_1,p_2) \in \cR^2$ for which the geodesic segment between $(p_1,p_2)$ and $(p_1+k_1,p_2+k_2)$ does not intersect $\cZ^2$, and then compute the sequence of crossings of this segment with the edges of the integral grid $\widetilde L^*$. That can be done in $O(n)$ time, using integer arithmetic only, with the standard Bresenham's line algorithm~\cite{5388473}.
\end{proof}

\begin{proposition}\label{prop:key-torus}
On a torus $\torus$, let $L$ be a canonical system of loops. Let $\Sigma$ be the surface obtained from $\torus$ by removing an open disk from the face of $L$. Let $C$ be a collection of quasi-geodesic closed walks in $L$. Then $i_{\torus}(C) = i_{\Sigma}(C)$.
\end{proposition}

\begin{proof}
We shall construct a closed curve $\gamma_c$ for every closed walk $c \in C$, homotopic to $c$ in $\Sigma$, in such a way that the collection $\{\gamma_c \;|\; c \in C\}$ is in general position and admits $i_\torus(C)$ crossings. First assume that every closed walk in $C$ is primitive. Let $L^*$ be the dual of $L$ on $\torus$. Endow $\torus$ with a flat metric for which the face of $L^*$ is isometric to the interior of a flat square. Identify the interior of $\Sigma$ with $\torus$ minus the vertex of $L^*$, such that the arcs of $\Sigma$ correspond to the loops of $L^*$. For every $c \in C$, use the assumption that $c$ is quasi-geodesic, and let $\gamma_c$ be a geodesic closed curve homotopic to $c$ in $\Sigma$. Every such $\gamma_c$ is simple since in the universal cover of $\torus$ the lifts of $\gamma_c$ are parallel geodesic lines, so they cannot intersect themself nor other lifts. Moreover, and without loss of generality, for every $c \neq d \in C$, the closed curves $\gamma_c$ and $\gamma_d$ do not overlap (otherwise, perturb them slightly). Also, they cross each other minimally among their homotopy classes in $\torus$. For otherwise, by the bigon criterion~\cite[Proposition~1.7]{fm-pmcg-12}, and since $\gamma_c$ and $\gamma_d$ are simple, they would form a bigon on $\torus$. The two sides of this bigon would be geodesic, a contradiction.

For the general case, for every $c \in C$, let $n \geq 1$ be such that $c$ is homotopic in $\torus$ to the $n^{th}$ power of a primitive closed curve. Then, and since $c$ is quasi-geodesic, $c$ is actually equal to the $n^{th}$ power of a primitive closed walk $\hat c$, where $\hat c$ is quasi-geodesic. Let $\widehat C := \{\hat c \;|\; c \in C \}$. By the previous paragraph, put the closed walks in $\widehat C$ in general position by homotopy in $\Sigma$, so that they intersect $i_{\torus}(\widehat C)$ many times. Then draw each $\gamma_c \in C$ in a neighborhood of $\hat c$ as in Figure~\ref{fig:non-primitive-curves}. The resulting collection $\{\gamma_c \;|\; c \in C\}$ is in general position and admits $i_{\torus}(C)$ crossings by Lemma~\ref{lem:degraaf}. 
\end{proof}

\begin{theorem}\label{thm:main-torus}
Let $M$ be a graph of size $m$ cellularly embedded on a torus $\torus$. Let $C$ be a collection of closed walks of length $n$ in $M$. One may construct a collection $C'$ of closed walks of length $O(mn)$ in $M$, homotopic to $C$, and a perturbation of $C'$ with $i_{\torus}(C)$ self-crossings, all in $O(mn + n \log n)$ time.
\end{theorem}

\begin{proof}
Let $L$ be obtained from $M$ by contracting a spanning tree, by merging any two edges that bound a bigon, and by deleting any edge that bounds a monogon. Then $L$ has only one vertex and either one quadrangular face, or two triangular faces. In the first case, $L$ is a canonical system of loops on $\torus$. In the second case, removing one loop makes $L$ a canonical system of loops. Push $C$ into a collection $C_1$ of closed walks of length $O(n)$ in $L$. Assume without loss of generality that every closed walk in $C_1$ is non-contractible (for otherwise, find the contractible ones, and handle them separately). Apply Lemma~\ref{lem:quasi-geodesic} to make the closed walks in $C_1$ quasi-geodesic in $O(n)$ time, keeping their length bounded by $O(n)$. Apply Theorem~\ref{fulek} to compute a minimal perturbation $\Gamma$ of $C_1$ in $O(n \log n)$ time. Then $\Gamma$ has $i_{\torus}(C)$ self-crossings by Lemma~\ref{lem:perturbations} and Proposition~\ref{prop:key-torus}. Recover from $C_1$ and $\Gamma$ a collection $C'$ of closed walks of length $O(mn)$ in $M$, and a perturbation of $C'$ without additional crossing, in $O(mn)$ time.
\end{proof}

\begin{corollary}
Let $M$ be a one-vertex graph cellularly embedded on a torus $\torus$. Let $C$ be a collection of closed walks of length $n$ in $M$. One may construct a collection $C'$ of closed walks of length $O(n)$ in $M$, homotopic to $C$, and a perturbation of $C'$ with $i_{\torus}(C)$ self-crossings, all in $O(n \log n)$ time. One may compute $i_{\torus}(C)$ at no extra cost.  
\end{corollary}

\begin{proof}
Construct a subgraph of $M$ by merging any two edges that bound a bigon, and by deleting any edge that bounds a monogon. This subgraph has size $O(1)$ by Euler formula. Modify $C$ accordingly. Apply Theorem~\ref{thm:main-torus} and Lemma~\ref{lem:compute-nb-crossings}.
\end{proof}

\paragraph{Acknowledgements.}{The author thanks Vincent Despré and Éric Colin de Verdière for many useful discussions.}

\bibliographystyle{plainurl}
\bibliography{bib}
\appendix

\section{Proof of Lemma~\ref{lem:minimal-position-2}}\label{app:minimal-position-2}

Let $S$ be an orientable surface of genus $g \geq 2$ without boundary. Let $\widetilde S$ be its universal cover, and $\pi : \widetilde S \to S$ be the covering map. Each of the following is classical, details and proof can be found in~\cite{untangling-graphs}, see also Farb and Margalit~\cite[Chapter~1]{fm-pmcg-12}.

\begin{lemma}[\cite{untangling-graphs}, Lemma~A.3]\label{lem:deck}
If $\tilde c : \cR \to \widetilde S$ is a lift of a non-contractible closed curve on $S$, then there is an automorphism $\lambda : \widetilde S \to \widetilde S$ such that $\pi \circ \lambda = \pi$ and $\lambda \circ \tilde c(t) = \tilde c(t+1)$ for every $t \in \cR$.
\end{lemma}

The universal cover $\widetilde S$ of $S$ can be compactified into a topological space $\widetilde S \cup \partial \widetilde S$ by adding a set $\partial \widetilde S$ of \emphdef{limit points}. The compactified space $\widetilde S \cup \partial \widetilde S$ is homeomorphic to a closed disk. Under this homeomorphism $\widetilde S$ is mapped to the interior of the disk, and $\partial \widetilde S$ is mapped to the boundary of the disk. Such a construction exists that satisfies each of the following:

\begin{lemma}[\cite{untangling-graphs}, Lemma~4.4]\label{lem:exist-limit-points}
If $\tilde c : \cR \to \widetilde S$ is a lift of a non-contractible closed curve on $S$, then $\lim_{+\infty} \tilde c$ and $\lim_{-\infty} \tilde c$ exist (in $\widetilde S \cup \partial \widetilde S$) and are distinct points of $\partial \widetilde S$.
\end{lemma}

\begin{lemma}[\cite{untangling-graphs}, Lemma~4.5]\label{lem:limit-points-homotopy}
Lift a homotopy $c \simeq d$ between non-contractible closed curves on $S$ to a homotopy $\tilde c \simeq \tilde d$ between lifts of $c$ and $d$. Then $\tilde c$ and $\tilde d$ have the same limit points.
\end{lemma}

\begin{lemma}[\cite{untangling-graphs}, Lemma~4.6]\label{lem:distinct-limit-points}
Consider lifts $\tilde c$ and $\tilde d$ of non-contractible closed curves on $S$. Assume that $\tilde c$ and $\tilde d$ intersect exactly once. Then the four limit points of $\tilde c$ and $\tilde d$ are pairwise distinct.
\end{lemma}

The following proof of Lemma~\ref{lem:minimal-position-2} reformulates classical arguments, see~\cite{r-ajccs-62} and~\cite[Lemma~4]{dl-cginc-19}.

\begin{proof}[Proof of Lemma~\ref{lem:minimal-position-2}]
We let $\rho : \cR \to \cR / \mathbb{Z}$ be the usual quotient map. Consider any two primitive closed curves $a$ and $b$ in general position on $S$. Fix a lift $\tilde a : \cR \to \widetilde S$ of $a \circ \rho$. Let $\tau : \cR \to \cR$ be defined by $\tau(t) = t+1$ for every $t \in \cR$. There is by Lemma~\ref{lem:deck} an automorphism $\lambda : \widetilde S \to \widetilde S$ such that $\lambda \circ \tilde a = \tilde a \circ \tau$. By Lemma~\ref{lem:exist-limit-points} $\tilde a$ admits two distinct limit points $p_1 \neq p_2 \in \partial \widetilde S$. Let $B$ contain the lifts of $b \circ \rho$ that admit two limit points $q_1,q_2 \in \partial \widetilde S$ for which $p_1,q_1,p_2,q_2$ are pairwise-distinct and in this order along $\partial \widetilde S$. The additive group $\mathbb{Z}^2$ acts on $B$ by mapping every $(i,j) \in \mathbb{Z}^2$ and every $\tilde b \in B$ to $\lambda^i \circ \tilde b \circ \tau^j \in B$. Let $N$ be the number of orbits of this action ($N$ does not depend on the choice of the lift $\tilde a$, but this is not needed in the proof). We have three claims that infer the result altogether.

Our first claim is that $N$ is smaller than or equal to the number of crossings between $a$ and $b$. To prove this claim observe that for any $\tilde b_1, \tilde b_2 \in B$, there are $s_1,t_1,s_2,t_2\in \cR$ such that $\tilde a(s_1) = \tilde b_1(t_1)$ and $\tilde a(s_2) = \tilde b_2(t_2)$. By definition, those points of $\widetilde S$ project on $S$ to crossing points between $a$ and $b$. Now, assuming that they project to the same crossing point, we shall prove that $\tilde b_1$ and $\tilde b_2$ belong to the same orbit in $B$. To do so first observe that, since $a$ and $b$ are primitive and in general position, we have $s_1 - s_2 \in \mathbb{Z}$ and $t_1 - t_2 \in \mathbb{Z}$. Let $i := s_2 - s_1$ and $j := t_1 - t_2$. Holds $\tilde b_2(t_2) = \tilde a(s_2) = \tilde a \circ \tau^i(s_1) = \lambda^i \circ \tilde a(s_1) = \lambda^i \circ \tilde b_1(t_1) = \lambda^i \circ \tilde b_1 \circ \tau^j(t_2)$. The maps $\tilde b_2$ and $\lambda^i \circ \tilde b_1 \circ \tau^j$ are both lifts of $b \circ \rho$, and they agree on $t_2$. Thus $\tilde b_2 = \lambda^i \circ \tilde b_1 \circ \tau^j$ by the uniquess part of the lifting property. That proves the first claim.

Our second claim is that if we assume that a lift of $a$ never crosses a lift of $b$ more than once, then $N$ is actually equal to the number of crossings between $a$ and $b$. To prove this claim, lift every such crossing between $a$ and $b$ to a crossing between $\tilde a$ and some lift $\tilde b$ of $b$. Every such lift $\tilde b$ belongs to $B$, since $\tilde a$ and $\tilde b$ cross only once so Lemma~\ref{lem:distinct-limit-points} applies. Now consider any two $\tilde b_1, \tilde b_2 \in B$ in the same orbit in $B$, and $s_1,t_1,s_2,t_2 \in \cR$ such that $\tilde a(s_1) = \tilde b(t_1)$ and $\tilde a(s_2) = \tilde b(t_2)$. We shall prove that those two crossing points in $\widetilde S$ project to the same crossing point in $S$, by proving that $s_1 - s_2 \in \mathbb{Z}$ and $t_1 - t_2 \in \mathbb{Z}$. There are $i, j \in \mathbb{Z}$ such that $\lambda^i \circ \tilde b_1 \circ \tau^j = \tilde b_2$. Then $\tilde a \circ \tau^{-i}(s_2) =  \lambda^{-i} \circ \tilde a(s_2) = \lambda^{-1} \circ \tilde b_2(t_2) = \tilde b_1 \circ \tau^j(t_2)$. Thus, and since $\tilde a$ and $\tilde b_1$ cross only once, we have $\tau^{-i}(s_2) = s_1$ and $\tau^j(t_2) = t_1$.

For stating our third claim, consider two homotopies $a \simeq a'$ and $b \simeq b'$, where $a'$ and $b'$ are closed curves in $S$. By the lifting property, the homotopy $a \simeq a'$ lifts to a homotopy $\tilde a \simeq \tilde a'$, where $\tilde a' : \cR \to \widetilde S$ is a lift of $a'$. The uniqueness part of the lifting property gives $\tilde a' = \lambda^{-1} \circ \tilde a' \circ \tau$ from the fact that $\tilde a = \lambda^{-1} \circ \tilde a \circ \tau$. In particular $\lambda \circ \tilde a' = \tilde a' \circ \tau$. As before, let $B'$ contain the lifts of $b'$ whose limit points alternate with those of $\tilde a'$ on $\partial \widetilde S$. The additive group $\mathbb{Z}^2$ acts on $B'$ by mapping every $(i,j) \in \mathbb{Z}^2$ and every $\tilde b' \in B'$ to $\lambda^i \circ \tilde b' \circ \tau^j \in B'$. Let $N'$ be the number of orbits of this action.

Our third claim is that $N \leq N'$. To prove this claim consider $\tilde b_1, \tilde b_2 \in B$ (if any). By the lifting property, the homotopy $b \simeq b'$  lifts to two homotopies $\tilde b_1 \simeq \tilde b_1'$ and $\tilde b_2 \simeq \tilde b_2'$, where $\tilde b_1', \tilde b_2' : \cR \to \widetilde S$ are lifts of $b'$. Lemma~\ref{lem:limit-points-homotopy} ensures that $\tilde b_1' \in B$ and $\tilde b_2' \in B$. If there are $i, j \in \mathbb{Z}$ such that $\lambda^i \circ \tilde b_1' \circ \tau^j = \tilde b_2'$, then the uniqueness part of the lifting property gives $\lambda^i \circ \tilde b_1 \circ \tau^j = \tilde b_2$.
\end{proof}

\section{Proof of Theorem~\ref{fulek}, following Fulek and Tóth}\label{app:fulek}

The result of Fulek and Tóth, Theorem~\ref{fulek}, has a different statement in their paper~\cite[Theorem~1]{fulek2020crossing}. In particular, the embedded graph lies in the plane, and only one closed walk is given as input. Nevertheless, their arguments easily extend to our setting. For completeness, we adapt their arguments~\cite[Section~3]{fulek2020crossing} to prove Theorem~\ref{fulek}. We insist that we do not introduce any new idea or result here.

In the whole section, we consider the setting of Theorem~\ref{fulek}: a collection $C$ of closed walks of length $n$ in an embedded graph $M$, without spur. We will prove Theorem~\ref{fulek} by constructing a minimal perturbation of $C$ in $O(m + n \log n)$ time. We assume that $M$ has no loop-edge nor any parallel-edges. For otherwise, construct $M'$ from $M$ by inserting two vertices in every edge. The collection $C$ becomes a collection $C'$ of closed walks of length $O(n)$ in $M'$, without spur. Construct a minimal perturbation $\Gamma'$ of $C'$. Easily derive from $\Gamma'$ a perturbation of $C$ with no more self-crossings than $\Gamma'$. 

We start with an initial perturbation $\Gamma$ of $C$. We call \emphdef{pebbles} the crossings of $\Gamma$ with the arcs of $\Sigma$, and we see the closed curves in $\Gamma$ as cycles of pebbles. Along any arc $a$ of $\Sigma$, the corresponding pebbles are linearly ordered along $a$. The problem is to reorder the pebbles along $a$, and that for every arc $a$ of $\Sigma$, so that in the end $\Gamma$ has minimum self-crossing.

We slightly modify the problem as follows. We partition every arc into a concatenation of sub-arcs. The partition of the pebbles into the sub-arcs and the linear orderings of the pebbles along the sub-arcs constitute an \emphdef{arrangement}. When reordering the pebbles, we are only allowed to reorder within the sub-arcs (pebbles cannot be exchanged between distinct sub-arcs). If reordering the pebbles this way can still make $\Gamma$ have minimum self-crossing, then the arrangement is \emphdef{valid}.

The \emphdef{split} operation modifies an arrangement as follows. It plays the role of the ``cluster expansion'' and ``pipe expansion'' operations of Fulek and Tóth. Let $D$ be a disk of $\Sigma$. 
Let $p_0, \dots, p_n$ for some $n \geq 1$ be the sub-arcs of the arrangement that bound $D$, in clockwise order, where $p_0$ is arbitrary. 
For every $1 \leq i \leq n$, let $X_i$ contain the pebbles of $p_0$ that are linked to a pebble of $p_i$ via an edge of $\Gamma$ that runs through $D$. If only one of the sets $X_1, \dots, X_n$ is not empty, then the split is not defined. Otherwise, the split cuts $p_0$ into sub-arcs $q_1, \dots, q_n$, in counter-clockwise order with respect to the orientation of $D$. Also, for every $1 \leq i \leq n$, the split places the pebbles in $X_i$ along $q_i$, in any order. Any newly-created sub-arc that does not contain a pebble is discarded.

\begin{lemma}\label{L:valid}
If an arrangement is valid before a split, then it is valid after the split.
\end{lemma}

\begin{proof}
Consider the arrangement before the split. In this arrangement, using the assumption that the arrangement is valid, reorder the pebbles within their sub-arcs so that $\Gamma$ has minimum self-crossing. Name $p$ the sub-arc to be split; $p$ is to be divided into sub-arcs $q_1, \dots, q_n$ for some $n \geq 2$. Assign an orientation to $p$, so that $q_1, \dots, q_n$ are in this order along $p$. For every $1 \leq i \leq n$, let $X_i$ contain the pebbles of $p$ that are to be placed in $q_i$. In our arrangement, consider the pebbles along $p$, and let $f : [k] \to [n]$ be the correspondence that sends every $l \in [k]$ to the unique $f(l) \in  [n]$ such that the $l$-th pebble along $p$ belongs to $X_{f(l)}$. If $f(l) > f(l+1)$ for some $1 \leq l < k$, then the $l$-th and the $l+1$-th pebbles can be swapped in our arrangement without increasing the number of self-crossings of $\Gamma$. Thus, without loss of generality, $f$ can be assumed non-decreasing. Then, the sub-arcs $q_1, \dots, q_n$ can be realized on $p$ (without touching to the order of the pebbles on $p$) so that, for every $i \in [n]$, the sub-arc $q_i$ contains the pebbles in $X_i$.
\end{proof}

The \emphdef{sub-arcs graph} of an arrangement is the graph whose vertices are the sub-arcs of the arrangement, and where two sub-arcs $p$ and $q$ are linked by an edge when there is a pebble in $p$ and a pebble in $q$ that are consecutive in some closed curve of $\Gamma$.

\begin{lemma}\label{L:endcase}
In a valid arrangement, if no split applies, then the pebbles can be re-ordered in $O(m + n \log n)$ time within their respective sub-arcs to make $\Gamma$ have minimum self-crossing.
\end{lemma}

\begin{proof}
If no split applies, then the sub-arcs graph of the arrangement is a collection of disjoint cycles. Deal with the cycles independently. The sub-arcs along a cycle $O$ contain the pebbles of a subset $\Gamma_0$ of closed curves from $\Gamma$. The closed curves in $\Gamma_0$ are all powers of the single closed curve that one obtains by making one loop around $O$. Re-order the pebbles within the sub-arcs of $O$, so that each closed curve in $\Gamma_0$ is as in Figure~\ref{fig:non-primitive-curves}, and so that any two closed curves do not intersect more than necessary. In the end, no re-ordering of the pebbles within their respective sub-arcs could induce less self-crossings of $\Gamma$, see e.g.~\cite[Theorems~6-7]{DEGRAAF1997134}. Thus, and since the arrangement is valid, $\Gamma$ has minimum self-crossing.
\end{proof}

\begin{proof}[Proof of Theorem~\ref{fulek}]
We assume, without loss of generality, that $M$ has no loop-edge nor any parallel-edges (see above). The overall algorithm is the following. Initialize an arrangement by considering an arbitrary perturbation of $C$. At this point, the arrangement is valid. Perform splits on it as long as possible. In the end, the arrangement is still valid by Lemma~\ref{L:valid}. Conclude by Lemma~\ref{L:endcase}. All there remains to do is to describe how to perform all the splits in $O(m + n \log n)$ total time.

We maintain the sub-arc graph $G$. Also, for every edge of $G$ between two sub-arcs $p$ and $q$, we maintain a list of the edges of $\Gamma$ between the pebbles in $p$ and the pebbles in $q$, and an integer that records the size of the list. Finally, we maintain the list of the sub-arcs that can be split.

Suppose we split a sub-arc $p_0$ along a disk $D$ of $\Sigma$. The $N \geq 2$ pebbles on $p_0$ are linked to pebbles in sub-arcs $p_1, \dots, p_k$ bounding $D$, for some $k \geq 2$. For each $i \in \{1,\dots,k\}$, let $X_i$ contain the pebbles of $p_0$ linked to a pebble in $p_i$, and let $N_i$ be the cardinality of $X_i$. ($X_i$ and $N_i$ are the list and the integer recorded by the corresponding edge of the sub-arc graph $G$.) Naively, we could range over $i \in \{1,\dots,k\}$, remove the pebbles in $X_i$ from the sub-arc $p_0$, place those pebbles on a new sub-arc, and update the sub-arc graph $G$, the lists and the integers, all in $O(N)$ time. Unfortunately, that would lead to a quadratic time algorithm. To overcome this issue, Fulek and Tóth crucially suggest to find some $j \in \{1,\dots,k\}$ maximizing $N_j$, and to range over $i \in \{1,\dots,k\} \setminus \{j\}$. In the end, the pebbles still remaining on $p_0$ are precisely the pebbles in $X_j$, we do not touch them. That takes $O(N - N_j)$ time. Summing this $N - N_j$ quantity over all the splits results in a total of $O(n \log n)$. Indeed, transfer a weight of $N - N_j$ to the pebbles in $X \setminus X_j$, by attributing a weight of $1$ to each pebble. After the split, each of those pebbles ends up in a sub-arc that contains no more than $N / 2$ pebbles. Thus, the sequence of splits attributes $O(\log n)$ weight to each pebble. There are $O(n)$ pebbles.
\end{proof}

\end{document}